\documentclass{article}

\usepackage[margin=1in]{geometry}
\usepackage{graphicx}
\usepackage{amssymb,amsmath,amsfonts, amsthm}
\usepackage{nicefrac}       
\usepackage{microtype}      
\usepackage{mathtools}
\usepackage[subtle]{savetrees}

\newtheorem{theorem}{Theorem}

\newtheorem{claim}[theorem]{Claim}

\newtheorem{fact}[theorem]{Fact}

\newcommand{\defeq}{\stackrel{\text{def}}{=}}
\newcommand{\E}{\mathbf{E}}
\DeclarePairedDelimiter{\cardinality}{\vert}{\vert}

\title{Adaptive Greedy Algorithms for Stochastic Set Cover Problems}
\author{Srinivasan Parthasarathy \\ IBM T.J. Watson Research Center \\ Yorktown Heights, NY 10598 \\ \texttt{spartha@us.ibm.com}}

\begin{document}
\maketitle

\begin{abstract}
  We study adaptive greedy algorithms for the problems of stochastic set cover with perfect and imperfect coverages.
  In stochastic set cover with perfect coverage, we are given a set of \textit{items} $\mathcal{A}$
  and a ground set $\mathcal{B}$. Evaluating
  an item reveals its \textit{state} which is a random subset of $\mathcal{B}$ drawn from the \textit{state distribution} of the item.
  Every element in $\mathcal{B}$ is assumed to be present in the state of some item with probability $1$.
  For this problem, we show that the adaptive greedy algorithm has an approximation ratio of $H(\cardinality{\mathcal{B}})$,
  the $\cardinality{\mathcal{B}}^{th}$ Harmonic number. In stochastic set cover
  with imperfect coverage, an element in the ground set need not be present in the state of any item.
  We show a reduction from this problem to the former problem; the adaptive greedy algorithm for
  the reduced instance has an approxiation ratio of $H(\cardinality{\mathcal{E}})$, where $\mathcal{E} \subseteq \mathcal{A} \times \mathcal{B}$
  is the set of pairs $(F, e)$ such that the state of item $F$ contains $e$ with positive probability.
\end{abstract}

\section{Introduction}
We study two variants of the stochastic set cover problem. In both these problems,
we are given a ground set $\mathcal{B}$ and a collection of \textit{items} $\mathcal{A}$.
Each $F \in \mathcal{A}$ is characterized by a
probability distribution $p_F: 2^{\mathcal{B}} \rightarrow [0,1]$. When item $F$ is evaluated,
it reveals a random subset $V(F) \sim p_F$ which is called the state of $F$.
Define $\forall (F,e)  \in \mathcal{A} \times \mathcal{B}: q_F(e) \defeq \Pr[e \in V(F)] = \sum_{V \ni e} p_F(V)$: this is
the probability with which the state of $F$ includes $e$.
We assume that the $q_F$'s are given but
$p_F$'s are not. Repeated evaluations of an item yields
the same state: in particular, evaluations of an item $F$ by two
distinct algorithms yields the same $V(F)$.  We assume states of items are mutually independent:
conditioning on the complete or partial state of items other than $F$
does not alter the state distribution of $F$. Let $C(F)$ denote the cost of evaluating $F$. Our goal is to evaluate a
subset of items $\mathcal{S}$ such that $\E[\sum_{F \in \mathcal{S}} C(F)]$ is minimized subject to
the following constraint:
\begin{gather}
  \Pr\left[\bigcup_{F \in \mathcal{S}} V(F) = \bigcup_{F \in \mathcal{A}} V(F)\right] = 1 \label{eqn:constraint}
\end{gather}
We will focus on \textit{oblivious} algorithms:
conditioning on the event that $F$ was evaluated by an oblivious algorithm does not
alter its state distribution.

The first of our two problems is stochastic set cover with perfect coverage. In this problem, we assume:
\begin{gather}
  \Pr\left[\bigcup_{F \in \mathcal{A}} V(F) = \mathcal{B} \right] = 1 \label{eqn:perfect}
\end{gather}
In other words, every element in $\mathcal{B}$ is included in the state of some item with probability $1$.

Our second problem is stochastic set cover with imperfect coverage which does not involve Assumption (\ref{eqn:perfect}).

\subsection{Related Work}
Stochastic set cover with perfect coverage has been studied in
~\cite{Deshpande:2016:AAS:2930058.2876506,goemans:stochasticcovering,Golovin:2011:AST:2208436.2208448,hellerstein:jair}.
Goemans and Vondr{\'a}k~\cite{goemans:stochasticcovering} study the adaptivity gap in stochastic covering problems\footnote{Stochastic set cover with perfect coverage
is referred to as stochastic set cover without item multiplicities in ~\cite{goemans:stochasticcovering}}.
An adaptive stochastic set cover algorithm is one which evaluates items iteratively, with
the choice of the item in the current iteration being influenced by the item states revealed in
the previous iterations. In contrast, a non-adaptive algorithm needs to choose its items upfront based
only on the knowledge of their state distributions in a manner which satisfies Constraint (\ref{eqn:constraint}). The
adaptivity gap is the maximum possible ratio between the expected cost of the optimal adaptive algorithm vs the optimal
non-adaptive algorithm. For stochastic set cover with perfect coverage, Goemans and Vondr{\'a}k~\cite{goemans:stochasticcovering} show
that the adaptivity gap can be as large as $\Omega(\cardinality{\mathcal{B}})$
and is bounded by $O(\cardinality{\mathcal{B}}^2)$, where $\cardinality{\mathcal{B}}$ is the size of the ground set.

The greedy algorithm for the deterministic set cover problem~\cite{Vazirani:2010:AA:1965254} can be generalized as
an adaptive algorithm for stochastic set cover with perfect coverage. This is the \texttt{Greedy}
algorithm described in Section \ref{sec:greedy}.
Golovin and Krause~\cite{Golovin:2011:AST:2208436.2208448} show that \texttt{Greedy}
has an approximation ratio of $(1 + \ln \cardinality{\mathcal{B}})^2$.
Deshpande, Hellerstein and Kletenik~\cite{Deshpande:2016:AAS:2930058.2876506}
show that \texttt{Greedy} has an approximation ratio of
$O(\max_{F \in \mathcal{A}} \cardinality{supp(p_F)}(1 + \ln \kappa))$ where
$supp(p_F)$ is the support of the state distribution function $p_F$, and $\kappa$ is the maximum cardinality of any state
which is realizable with positive probability. Hellerstein and Kletenik~\cite{hellerstein:jair}
show that \texttt{Greedy} has an approximation ratio of $O(1 + \frac{\ln \kappa}{\eta_E})$
where $\eta_E$ is the smallest expected number of new elements that are covered
when \texttt{Greedy} evaluates an item. In this paper, we show that \texttt{Greedy} has an approximation ratio of
$H(\cardinality{\mathcal{B}}) \leq 1 + \log{\cardinality{\mathcal{B}}}$ which strictly improves the bound of
~\cite{Golovin:2011:AST:2208436.2208448} and brings the approximation ratio closer to that of the greedy algorithm
for the deterministic set cover problem~\cite{chvatal:setcover,johnson:setcover,lovasz:setcover,stein:setcover}.

For the problem of stochastic set cover with imperfect coverage, we are not aware of any known approximation algorithms
prior to this work.

\section{Stochastic Set Cover with Perfect Coverage: \texttt{Greedy} Algorithm}
\label{sec:greedy}
\texttt{Greedy} evaluates items sequentially. Let $F_i$ denote the $i^{th}$
item evaluated by \texttt{Greedy}. Suppose after $i-1$ item evaluations,
$\mathcal{A} \setminus \{F_1, F_2, \ldots, F_{i-1}\} = \mathcal{A}_i$
and $\mathcal{B} \setminus (\bigcup_{j = 1}^{i-1} V(F_j)) = \mathcal{B}_i$; we
refer to $(\mathcal{A}_i, \mathcal{B}_i)$ as the $i^{th}$ \textit{residual system}, and
say that the event $\Gamma(i; \mathcal{A}_i, \mathcal{B}_i)$
occured. Define:
\begin{gather}
  \forall F \in \mathcal{A}_i \text{ s.t. } \sum_{e \in \mathcal{B}_i} q_F(e) > 0: unitprice_i(F, \mathcal{A}_i, \mathcal{B}_i) \defeq \frac{C(F)}{\sum_{e \in \mathcal{B}_i} q_F(e)} \label{def:unitpricei}
\end{gather}

\texttt{Greedy} terminates if $\mathcal{B}_i = \Phi$.
Otherwise \texttt{Greedy} evaluates \newline
$F_i = \arg\min_{F \in \mathcal{A}_i \text{ s.t. } \sum_{e \in \mathcal{B}_i} q_F(e) > 0} unitprice_i(F, \mathcal{A}_i, \mathcal{B}_i)$ and
computes the new residual system $(\mathcal{A}_{i+1}, \mathcal{B}_{i+1})$. In the latter case,
$\forall e \in \mathcal{B}_i \setminus \mathcal{B}_{i+1}$, we say $e$ is covered by $F_i$ in \texttt{Greedy} and denote
this by $F_i \succ e$.

For all $F \in G$, we define $unitprice(F)$ to be the unitprice of $F$ just before its evaluation. Further:
\begin{gather}
  \forall e \text{ s.t. } F_i \succ e: price(e) \defeq unitprice(F_i) \label{eqn:price} \\
  gCov(F_i) \defeq \cardinality{\mathcal{B}_i \setminus \mathcal{B}_{i+1}} = \sum_{e \in \mathcal{B}_i} \mathbf{1}_{F_i \succ e} \label{eqn:gcov}
\end{gather}

If an element $e$ is not covered by any item when \texttt{Greedy} terminates, we set $price(e) \defeq 0$.
Let $G$ denote the set of all items evaluated by \texttt{Greedy}. The termination condition of \texttt{Greedy} guarantees
Constraint (\ref{eqn:constraint}). We now analyze the cost of \texttt{Greedy}.\footnote{The quantity defined in Eqn (\ref{def:unitpricei}) is parameterized by the $i^{th}$ residual system and is a
fixed number and not a random variable. $\forall F \in G, ~unitprice(F)$ is a random variable. If we condition on the $i^{th}$ residual system and the
event that $F_i = F$, then $unitprice(F)$ becomes a fixed number. The denominator in the R.H.S. of Eqn (\ref{def:unitpricei}) is the expected number of elements that will be
covered if $F_i = F$.}

\subsection{Analysis of the \texttt{Greedy} Algorithm}
Let $(\mathcal{A}, \mathcal{B})$ denote the residual system at the start of \texttt{Greedy}. The following fact is
used repeatedly in our analysis.
\begin{fact}
  \label{fact:oblivious}
  Both \texttt{Greedy} and \texttt{optimal} are oblivious algorithms, and item states are independent. This implies
  that the conditional state distribution of an item $F$ -- given that $F$ was evaluated by
  \texttt{Greedy} or \texttt{optimal} or both or neither, and given (partial or complete) information about the state of
  any items other than $F$ -- is the same as the unconditional state distribution of $F$.
\end{fact}

\begin{claim}
\label{claim:grecost}
\begin{gather}
  \sum_{F \in \mathcal{A}} \Pr[F \in G]C(F) = \sum_{e \in \mathcal{B}} \E[price(e)] \label{eqn:claimgrecost}
\end{gather}
\end{claim}
\begin{proof}
  \begin{gather}
    \sum_{e \in \mathcal{B}} \E[price(e)] = \sum_{e \in \mathcal{B}} \sum_{F \in \mathcal{A}}\Pr[F \succ e]\E[price(e) \vert F \succ e]
    = \sum_{e \in \mathcal{B}} \sum_{F \in \mathcal{A}} \Pr[F \in G] \Pr[F \succ e \vert F \in G] \E[price(e) \vert F \succ e] \nonumber \\
    = \sum_{F \in \mathcal{A}} \Pr[F \in G] \sum_{e \in \mathcal{B}} \Pr[F \succ e \vert F \in G] \E[price(e) \vert F \succ e]
    = \sum_{F \in \mathcal{A}} \Pr[F \in G] \sum_{e \in \mathcal{B}} \E[\mathbf{1}_{F \succ e} price(e) \vert F \in G] \nonumber \\
    = \sum_{F \in \mathcal{A}} \Pr[F \in G] \sum_{e \in \mathcal{B}} ~~~\sum_{\substack{\Gamma(i; \mathcal{A}_i, \mathcal{B}_i) \\ \text{s.t. } F_i = F}} \Pr[\Gamma(i; \mathcal{A}_i, \mathcal{B}_i) \vert F \in G] \E[\mathbf{1}_{F \succ e} price(e) \vert \Gamma(i; \mathcal{A}_i, \mathcal{B}_i)] \nonumber \\
    = \sum_{F \in \mathcal{A}} \Pr[F \in G] \sum_{\substack{\Gamma(i; \mathcal{A}_i, \mathcal{B}_i) \\ \text{s.t. } F_i = F}} \Pr[\Gamma(i; \mathcal{A}_i, \mathcal{B}_i) \vert F \in G] \sum_{e \in \mathcal{B}} \E[\mathbf{1}_{F \succ e} price(e) \vert \Gamma(i; \mathcal{A}_i, \mathcal{B}_i)] \nonumber \\
    = \sum_{F \in \mathcal{A}} \Pr[F \in G] \sum_{\substack{\Gamma(i; \mathcal{A}_i, \mathcal{B}_i) \\ \text{s.t. } F_i = F}} \Pr[\Gamma(i; \mathcal{A}_i, \mathcal{B}_i) \vert F \in G] unitprice_i(F, \mathcal{A}_i, \mathcal{B}_i) \E\left[\left(\sum_{e \in \mathcal{B}_i} \mathbf{1}_{F \succ e} \right) \bigg\vert \Gamma(i; \mathcal{A}_i, \mathcal{B}_i) \right] \nonumber \\
    = \sum_{F \in \mathcal{A}} \Pr[F \in G] \sum_{\substack{\Gamma(i; \mathcal{A}_i, \mathcal{B}_i) \\ \text{s.t. } F_i = F}} \Pr[\Gamma(i; \mathcal{A}_i, \mathcal{B}_i) \vert F \in G] unitprice_i(F, \mathcal{A}_i, \mathcal{B}_i) \E[gCov(F) \vert \Gamma(i; \mathcal{A}_i, \mathcal{B}_i)] \nonumber \\
    = \sum_{F \in \mathcal{A}} \Pr[F \in G] \sum_{\substack{\Gamma(i; \mathcal{A}_i, \mathcal{B}_i) \\ \text{s.t. } F_i = F}} \Pr[\Gamma(i; \mathcal{A}_i, \mathcal{B}_i) \vert F \in G] C(F) \nonumber = \sum_{F \in \mathcal{A}} \Pr[F \in G] C(F) \nonumber
  \end{gather}
\end{proof}

Let $O$ be the set of items evaluated by the optimal algorithm (\texttt{optimal}).
For an item $F$, we define:
\begin{gather}
oRev(F) \defeq
  \begin{cases}
    unitprice(F) gCov(F) & \forall F \in O \cap G \\
    C(F) & \forall F \in O \setminus G \\
    0 & \forall F \not\in O
  \end{cases}
\label{eqn:orev}
\end{gather}

\begin{claim}
\label{claim:ocosteqorev}
\begin{gather}
  \sum_{F \in \mathcal{A}} \E[oRev(F)] = \E\left[\sum_{F \in O} C(F)\right] \label{eqn:ocosteqorev}
\end{gather}
\end{claim}
\begin{proof}
  \begin{gather}
    \sum_{F \in \mathcal{A}} \E[oRev(F)] = \sum_{F \in \mathcal{A}} \Pr[F \in O \cap G]\E[oRev(F) \vert F \in O \cap G] + \Pr[F \in O \setminus G]C(F) \nonumber \\
    = \sum_{F \in \mathcal{A}} \Pr[F \in O \cap G] \sum_{\substack{\Gamma(i; \mathcal{A}_i, \mathcal{B}_i) \\ \text{s.t. } F_i = F}} \Pr[\Gamma(i; \mathcal{A}_i, \mathcal{B}_i) \vert F \in O \cap G]\E[oRev(F) \vert \Gamma(i; \mathcal{A}_i, \mathcal{B}_i) \wedge F \in O \cap G] + \Pr[F \in O \setminus G]C(F)  \nonumber \\
    = \sum_{F \in \mathcal{A}} \Big(\big\{\Pr[F \in O \cap G] \sum_{\substack{\Gamma(i; \mathcal{A}_i, \mathcal{B}_i) \\ \text{s.t. } F_i = F}} \Pr[\Gamma(i; \mathcal{A}_i, \mathcal{B}_i) \vert F \in O \cap G] \cdot \nonumber \\
    unitprice_i(F, \mathcal{A}_i, \mathcal{B}_i) \E[gCov(F) \vert \Gamma(i; \mathcal{A}_i, \mathcal{B}_i) \wedge F \in O \cap G]\big\}
    + \Pr[F \in O \setminus G]C(F)\Big) \nonumber \\
    \stackrel{\text{Fact \ref{fact:oblivious}}}{=}
    \sum_{F \in \mathcal{A}} \Pr[F \in O \cap G] \sum_{\substack{\Gamma(i; \mathcal{A}_i, \mathcal{B}_i) \\ \text{s.t. } F_i = F}} \Pr[\Gamma(i; \mathcal{A}_i, \mathcal{B}_i) \vert F \in O \cap G] C(F) + \Pr[F \in O \setminus G]C(F) = \E\left[\sum_{F \in O} C(F)\right] \nonumber
  \end{gather}
\end{proof}

Let $e_1, e_2, \ldots e_{\cardinality{\mathcal{B}}}$ be an ordering of elements in $\mathcal{B}$ such that
elements covered earlier precede the elements covered later, with ties broken arbitrarily.
\begin{claim}
\label{claim:oreveqcost}
\begin{gather}
  \forall (i; \mathcal{A}_i, \mathcal{B}_i), \forall j > \cardinality{\mathcal{B} \setminus \mathcal{B}_i} \text{ s.t. } \Pr[F_i \succ e_j \vert \Gamma(i; \mathcal{A}_i, \mathcal{B}_i)] > 0, \forall F \in \mathcal{A}_{i+1}: \nonumber \\
  \E[oRev(F) \vert \Gamma(i; \mathcal{A}_i, \mathcal{B}_i) \wedge F_i \succ e_j] = \Pr[F \in O \vert \Gamma(i; \mathcal{A}_i, \mathcal{B}_i) \wedge F_i \succ e_j]C(F) \label{eqn:oreveqcost}
\end{gather}
\end{claim}
\begin{proof}
  \begin{gather}
    \forall (i; \mathcal{A}_i, \mathcal{B}_i), \forall j > \cardinality{\mathcal{B} \setminus \mathcal{B}_i} \text{ s.t. } \Pr[F_i \succ e_j \vert \Gamma(i; \mathcal{A}_i, \mathcal{B}_i)] > 0, \forall F \in \mathcal{A}_{i+1}: \nonumber \\
    \E[oRev(F) \vert \Gamma(i; \mathcal{A}_i, \mathcal{B}_i) \wedge F_i \succ e_j] \nonumber \\
    = \Pr[F \in O \setminus G \vert \Gamma(i; \mathcal{A}_i, \mathcal{B}_i) \wedge F_i \succ e_j]C(F) \nonumber \\
    + \Pr[F \in O \cap G \vert \Gamma(i; \mathcal{A}_i, \mathcal{B}_i) \wedge F_i \succ e_j]\E[oRev(F) \vert \Gamma(i; \mathcal{A}_i, \mathcal{B}_i) \wedge F_i \succ e_j \wedge F \in O \cap G] \nonumber \\
    =  \Pr[F \in O \setminus G \vert \Gamma(i; \mathcal{A}_i, \mathcal{B}_i) \wedge F_i \succ e_j]C(F)
    + \Pr[F \in O \cap G \vert \Gamma(i; \mathcal{A}_i, \mathcal{B}_i) \wedge F_i \succ e_j] \cdot \nonumber \\
     \bigg(\sum_{\substack{\Gamma(k; \mathcal{A}_k, \mathcal{B}_k) \\
     \text{ s.t. } k > i \wedge F_k = F}} \Pr[\Gamma(k; \mathcal{A}_k, \mathcal{B}_k) \vert \Gamma(i; \mathcal{A}_i, \mathcal{B}_i) \wedge F_i \succ e_j \wedge F \in O \cap G]
    unitprice_k(F, \mathcal{A}_k, \mathcal{B}_k) \cdot \nonumber \\
    \E[gCov(F) \vert \Gamma(k; \mathcal{A}_k, \mathcal{B}_k) \wedge \Gamma(i; \mathcal{A}_i, \mathcal{B}_i) \wedge F_i \succ e_j \wedge F \in O \cap G] \bigg) \nonumber \\
    \stackrel{\text{Fact \ref{fact:oblivious}}}{=}  \Pr[F \in O \setminus G \vert \Gamma(i; \mathcal{A}_i, \mathcal{B}_i) \wedge F_i \succ e_j]C(F)
    + \Pr[F \in O \cap G \vert \Gamma(i; \mathcal{A}_i, \mathcal{B}_i) \wedge F_i \succ e_j]C(F) \label{eqn:state} \\
    = \Pr[F \in O \vert \Gamma(i; \mathcal{A}_i, \mathcal{B}_i) \wedge F_i \succ e_j]C(F) \nonumber
  \end{gather}
\end{proof}

Let $F \gets e$ denote the event that element $e$ was covered by item $F$ in \texttt{optimal}.
We assume WLOG that \texttt{optimal} is also a sequential algorithm like \texttt{Greedy}. Elements removed from the residual
system of \texttt{optimal} due to the evaluation of item $F$ are considered to be covered by $F$ in \texttt{optimal}.
\begin{claim}
\label{claim:ocovleqden}
\begin{gather}
  \forall (i; \mathcal{A}_i, \mathcal{B}_i), \forall j > \cardinality{\mathcal{B} \setminus \mathcal{B}_i} \text{ s.t. } \Pr[F_i \succ e_j \vert \Gamma(i; \mathcal{A}_i, \mathcal{B}_i)] > 0, \forall F \in \mathcal{A}_{i+1}: \nonumber \\
  \E\left[\sum_{e \in \mathcal{B}_i} \mathbf{1}_{F \gets e} \bigg\vert \Gamma(i; \mathcal{A}_i, \mathcal{B}_i) \wedge F_i \succ e_j \right] \leq \Pr[F \in O \vert \Gamma(i; \mathcal{A}_i, \mathcal{B}_i) \wedge F_i \succ e_j] \sum_{e \in \mathcal{B}_i} q_F(e) \label{eqn:ocovleqden}
\end{gather}
\end{claim}
\begin{proof}
  Consider any $(i; \mathcal{A}_i, \mathcal{B}_i), j > \cardinality{\mathcal{B} \setminus \mathcal{B}_i} \text{ s.t. } \Pr[F_i \succ e_j \vert \Gamma(i; \mathcal{A}_i, \mathcal{B}_i)] > 0$.
  $F \gets e \implies e \in V(F)$. Hence,
  \begin{gather}
    \forall F \in \mathcal{A}_{i+1}: ~\E\left[\sum_{e \in \mathcal{B}_i} \mathbf{1}_{F \gets e} \bigg\vert \Gamma(i; \mathcal{A}_i, \mathcal{B}_i) \wedge F_i \succ e_j \right] \nonumber \\
    = \Pr[F \in O \vert \Gamma(i; \mathcal{A}_i, \mathcal{B}_i) \wedge F_i \succ e_j] \E\left[\sum_{e \in \mathcal{B}_i} \mathbf{1}_{F \gets e} \bigg\vert \Gamma(i; \mathcal{A}_i, \mathcal{B}_i) \wedge F_i \succ e_j \wedge F \in O \right] \nonumber \\
    \leq \Pr[F \in O \vert \Gamma(i; \mathcal{A}_i, \mathcal{B}_i) \wedge F_i \succ e_j] \E\left[\sum_{e \in \mathcal{B}_i} \mathbf{1}_{e \in V(F)} \bigg\vert \Gamma(i; \mathcal{A}_i, \mathcal{B}_i) \wedge F_i \succ e_j \wedge F \in O \right] \nonumber \\
    \stackrel{\text{Fact \ref{fact:oblivious}}}{=} \Pr[F \in O \vert \Gamma(i; \mathcal{A}_i, \mathcal{B}_i) \wedge F_i \succ e_j] \sum_{e \in \mathcal{B}_i} q_F(e) \label{eqn:pcoveqden}
  \end{gather}
\end{proof}

\begin{fact}
\label{fact:atleastmin}
For non-negative numbers $a_1, \ldots, a_{\ell}$, and $b_1, \ldots, b_{\ell}$ s.t. $\sum_{t = 1}^{\ell} b_t > 0$:
\begin{gather}
\frac{\sum_{t} a_t}{\sum_{t} b_t} \geq \min_{t \vert b_t > 0} \frac{a_t}{b_t} \label{eqn:atleastmin}
\end{gather}
\end{fact}

\begin{theorem}
  \label{theo:greedy}
  Algorithm \texttt{Greedy} has an approximation ratio of $H(\cardinality{\mathcal{B}})$.
\end{theorem}
\begin{proof}
  Suppose $(\mathcal{A}_i, \mathcal{B}_i)$ is the $i^{th}$ residual system.
  Consider any $j > \cardinality{\mathcal{B} \setminus \mathcal{B}_i}$ s.t. $\Pr[F_i \succ e_j \vert \Gamma(i; \mathcal{A}_i, \mathcal{B}_i)] > 0$.
  We have:
  \begin{gather}
    \E[\sum_{e \in \mathcal{B}_i} \mathbf{1}_{F_i \gets e} \vert \Gamma(i; \mathcal{A}_i, \mathcal{B}_i) \wedge F_i \succ e_j] \nonumber \\
    = \Pr[F_i \in O \vert \Gamma(i; \mathcal{A}_i, \mathcal{B}_i) \wedge F_i \succ e_j] \E[\sum_{e \in \mathcal{B}_i} \mathbf{1}_{F_i \gets e} \vert \Gamma(i; \mathcal{A}_i, \mathcal{B}_i) \wedge F_i \succ e_j \wedge F_i \in O] \nonumber \\
    \leq \Pr[F_i \in O \vert \Gamma(i; \mathcal{A}_i, \mathcal{B}_i) \wedge F_i \succ e_j] \E[gCov(F_i) \vert \Gamma(i; \mathcal{A}_i, \mathcal{B}_i) \wedge F_i \succ e_j \wedge F_i \in O] \label{eqn:moreingreedy}
  \end{gather}
  Eqn (\ref{eqn:moreingreedy}) follows from the fact that the number of elements in $\mathcal{B}_i$ covered by $F_i$ in \texttt{optimal}
  cannot exceed the corresponding number in \texttt{Greedy}.
  Suppose $\E[\sum_{e \in \mathcal{B}_i} \mathbf{1}_{F_i \gets e} \vert \Gamma(i; \mathcal{A}_i, \mathcal{B}_i) \wedge F_i \succ e_j] > 0$.
  We now have:
  \begin{gather}
    \frac{\E[oRev(F_i) \vert \Gamma(i; \mathcal{A}_i, \mathcal{B}_i) \wedge F_i \succ e_j]}{\E[\sum_{e \in \mathcal{B}_i} \mathbf{1}_{F_i \gets e} \vert \Gamma(i; \mathcal{A}_i, \mathcal{B}_i) \wedge F_i \succ e_j]} \nonumber \\
    \stackrel{\text{Eqn (\ref{eqn:moreingreedy})}}{\geq} \frac{\Pr[F_i \in O \vert \Gamma(i; \mathcal{A}_i, \mathcal{B}_i) \wedge F_i \succ e_j] \E[unitprice(F_i) gCov(F_i) \vert \Gamma(i; \mathcal{A}_i, \mathcal{B}_i) \wedge F_i \succ e_j \wedge F_i \in O]}{\Pr[F_i \in O \vert \Gamma(i; \mathcal{A}_i, \mathcal{B}_i) \wedge F_i \succ e_j] \E[gCov(F_i) \vert \Gamma(i; \mathcal{A}_i, \mathcal{B}_i) \wedge F_i \succ e_j \wedge F_i \in O]} \nonumber \\
    = unitprice_i(F_i; \mathcal{A}_i, \mathcal{B}_i) \nonumber \\
    = \E[price(e_j) \vert \Gamma(i; \mathcal{A}_i, \mathcal{B}_i) \wedge F_i \succ e_j] \label{eqn:atleastforfi}
  \end{gather}

  Now consider any $F \in \mathcal{A}_{i+1}$ s.t.
  $\E[\sum_{e \in \mathcal{B}_i} \mathbf{1}_{F \gets e} \vert \Gamma(i; \mathcal{A}_i, \mathcal{B}_i) \wedge F_i \succ e_j] > 0$. We have:
  \begin{gather}
    \frac{\E[oRev(F) \vert \Gamma(i; \mathcal{A}_i, \mathcal{B}_i) \wedge F_i \succ e_j]}{\E[\sum_{e \in \mathcal{B}_i} \mathbf{1}_{F \gets e} \vert \Gamma(i; \mathcal{A}_i, \mathcal{B}_i) \wedge F_i \succ e_j]}
    \stackrel{\text{Eqn (\ref{eqn:oreveqcost}), (\ref{eqn:ocovleqden})}}{\geq}
    \frac{C(F)}{\sum_{e \in \mathcal{B}_i} q_F(e)} \geq
    \E[price(e_j) \vert \Gamma(i; \mathcal{A}_i, \mathcal{B}_i) \wedge F_i \succ e_j] \label{eqn:atleastforf}
  \end{gather}
  Eqn (\ref{eqn:atleastforf}) holds because, in each iteration, \texttt{Greedy} chooses the item which minimizes unitprice.

  We now have:
  \begin{gather}
    \forall (i; \mathcal{A}_i, \mathcal{B}_i), \forall j > \cardinality{\mathcal{B} \setminus \mathcal{B}_i} \text{ s.t. } \Pr[F_i \succ e_j \vert \Gamma(i; \mathcal{A}_i, \mathcal{B}_i)] > 0: \nonumber \\
    \frac{\sum_{F \in \mathcal{A}} \E[oRev(F) \vert \Gamma(i; \mathcal{A}_i, \mathcal{B}_i) \wedge F_i \succ e_j]}{\cardinality{\mathcal{B}} - j + 1}
    \geq \frac{\sum_{F \in \mathcal{A}_i} \E[oRev(F) \vert \Gamma(i; \mathcal{A}_i, \mathcal{B}_i) \wedge F_i \succ e_j]}{\cardinality{\mathcal{B}_i}} \nonumber \\
    = \frac{\E[oRev(F_i) \vert \Gamma(i; \mathcal{A}_i, \mathcal{B}_i) \wedge F_i \succ e_j] + \sum_{F \in \mathcal{A}_{i+1}} \E[oRev(F) \vert \Gamma(i; \mathcal{A}_i, \mathcal{B}_i) \wedge F_i \succ e_j]}{\E[\sum_{e \in \mathcal{B}_i} \mathbf{1}_{F_i \gets e} \vert \Gamma(i; \mathcal{A}_i, \mathcal{B}_i) \wedge F_i \succ e_j] + \sum_{F \in \mathcal{A}_{i + 1}} \E[\sum_{e \in \mathcal{B}_i} \mathbf{1}_{F \gets e} \vert \Gamma(i; \mathcal{A}_i, \mathcal{B}_i) \wedge F_i \succ e_j]} \nonumber \\
    \stackrel{\text{Eqn (\ref{eqn:atleastmin}), (\ref{eqn:atleastforfi}), (\ref{eqn:atleastforf})}}{\geq} \E[price(e_j) \vert \Gamma(i; \mathcal{A}_i, \mathcal{B}_i) \wedge F_i \succ e_j] \label{eqn:atleastprice}
  \end{gather}
  Thus, we have:
  \begin{gather}
    \frac{\E[\sum_{F \in O} C(F)]}{\cardinality{\mathcal{B}} - j + 1} \stackrel{\text{Eqn (\ref{eqn:ocosteqorev})}}{=} \frac{\sum_{F \in \mathcal{A}} \E[oRev(F)]}{\cardinality{\mathcal{B}} - j + 1} \nonumber \\
    = \sum_{\substack{\Gamma(i; \mathcal{A}_i, \mathcal{B}_i) \\ \text{ s.t. } \Pr[F_i \succ e_j \vert \Gamma(i; \mathcal{A}_i, \mathcal{B}_i)] > 0}}
    \Pr[\Gamma(i; \mathcal{A}_i, \mathcal{B}_i) \wedge F_i \succ e_j]
    \frac{\sum_{F \in \mathcal{A}} \E[oRev(F) \vert \Gamma(i; \mathcal{A}_i, \mathcal{B}_i) \wedge F_i \succ e_j]}{\cardinality{\mathcal{B}} - j + 1} \nonumber \\
    \stackrel{\text{Eqn (\ref{eqn:atleastprice})}}{\geq}
    \sum_{\substack{\Gamma(i; \mathcal{A}_i, \mathcal{B}_i) \\ \text{ s.t. } \Pr[F_i \succ e_j \vert \Gamma(i; \mathcal{A}_i, \mathcal{B}_i)] > 0}}
    \Pr[\Gamma(i; \mathcal{A}_i, \mathcal{B}_i) \wedge F_i \succ e_j]
    \E[price(e_j) \vert \Gamma(i; \mathcal{A}_i, \mathcal{B}_i) \wedge F_i \succ e_j] = \E[price(e_j)] \label{eqn:oneprice}
  \end{gather}
  Hence,
  \begin{gather}
    H(\cardinality{\mathcal{B}})\E[\sum_{F \in O} C(F)] = \sum_{j = 1}^{\cardinality{\mathcal{B}}} \frac{\E[\sum_{F \in O} C(F)]}{\cardinality{\mathcal{B}} - j + 1}
    ~~\stackrel{\text{Eqn (\ref{eqn:oneprice})}}{\geq}~~ \sum_{j = 1}^{\cardinality{\mathcal{B}}} \E[price(e_j)]
    ~~\stackrel{\text{Eqn (\ref{eqn:claimgrecost})}}{=}~~ \sum_{F \in \mathcal{A}}\Pr[F \in G]C(F) \nonumber
  \end{gather}
  Hence the theorem holds.
\end{proof}

\section{Stochastic Set Cover with Imperfect Coverage}
We now describe a reduction from
an instance $\mathcal{I}$ of stochastic set cover with imperfect coverage
to an instance $\mathcal{J}$ of stochastic set cover with perfect coverage.
Let $\mathcal{A}$, $\mathcal{B}$, $p_F$, and $q_F$ be the items,
ground set, state distributions, and marginal distributions in $\mathcal{I}$.\footnote{Recall that $q_F$'s are given but $p_F$'s are not assumed to be known.}
Consider the bipartite graph $\mathcal{G} = (\mathcal{A}, \mathcal{B}, \mathcal{E})$ where $\mathcal{E}$
consists of all edges $(F, e) \in \mathcal{A} \times \mathcal{B}$ s.t. $q_F(e) > 0$.
We refer to $\mathcal{G}$ as the \textit{bipartite graph induced by} $\mathcal{I}$.
Define function $\mu_F: supp(p_F) \rightarrow 2^{\mathcal{E}}$ s.t.
$\forall B \in supp(p_F), ~\mu(B) = \{(F_1, e) ~\vert~ (F_1, e) \in \mathcal{E} \wedge (F_1 = F \vee e \in B) \}$.
Suppose the evaluation of item $F$ yields the state $B$ in instance $\mathcal{I}$;
$\mu_F(B)$ can then be readily computed from $\mathcal{G}$.

We create instance $\mathcal{J}$ as follows: the set of items in $\mathcal{J}$ is $\mathcal{A}$ and the ground set is $\mathcal{E}$.
Suppose the evaluation of item $F$ yields the state $B$ in $\mathcal{I}$; we then treat the state of
item $F$ in $\mathcal{J}$ to be $\mu_F(B)$.
This construction implies that the marginal distributions $m_F$ in $\mathcal{J}$
are as follows:
\begin{gather}
\forall F \in \mathcal{A}, \forall (F_1, e) \in \mathcal{E}, ~ m_F((F_1, e)) =
  \begin{cases}
    1 & \text{ if } F_1 = F \\
    q_F(e) & \text{ if } F_1 \neq F
  \end{cases}
\label{eqn:marginalJ}
\end{gather}
This completes the description of instance $\mathcal{J}$.

\begin{theorem}
  \label{theo:imperfect}
  The solution produced by \texttt{Greedy} on $\mathcal{J}$ yields a solution to $\mathcal{I}$
  with an approximation ratio $H(\cardinality{\mathcal{E}})$,
  where $\mathcal{E}$ is the set of edges in the bipartite graph induced by $\mathcal{I}$.\footnote{$\cardinality{\mathcal{E}} \leq \cardinality{\mathcal{A}}\cardinality{\mathcal{B}}$.}
\end{theorem}
\begin{proof}
  It is easy to verify that $S \subseteq \mathcal{A}$ is a valid solution for $\mathcal{I}$ if and only if it is
  a valid solution for $\mathcal{J}$. Our reduction guarantees that $\forall (F, e) \in \mathcal{E}$,
  the state of $F$ in $\mathcal{J}$ contains $(F, e)$ with probability $1$.
  Hence, Assumption (\ref{eqn:perfect}) is satisfied which makes $\mathcal{J}$ an instance of stochastic set cover
  with perfect coverage. The theorem now follows by noting that the size of the ground set in $\mathcal{J}$
  is $\cardinality{\mathcal{E}}$.
\end{proof}

\section{Open Problem}
The greedy algorithm for deterministic set cover problem has an approximation ratio of $H(\ell)$, where $\ell$
is the maximum cardinality of an item~\cite{chvatal:setcover}. Analogously,
does \texttt{Greedy} have an approximation ratio of $H(n)$
for stochastic set cover with perfect coverage where $n$ is the
maximum cardinality of any state which is realizable with positive probability?

\section*{Acknowledgements}
We are grateful to Lisa Hellerstein (NYU) for several useful comments on various drafts of this work, and for
discovering and clearly describing an error in \cite{Liu:2008:NAS:1376616.1376633} which lead to this work.
We are also grateful to Arpit Agarwal, Sepehr Assadi, and Sanjeev Khanna (U. Penn) for
helpful comments on earlier drafts of this work.

\bibliography{greedystochastic}{}

\begin{thebibliography}{10}

\bibitem{chvatal:setcover}
V.~Chv\'{a}tal.
\newblock A greedy heuristic for the set-covering problem.
\newblock {\em Mathematics of Operations Research}, 4(3):233--235, 1979.

\bibitem{Deshpande:2016:AAS:2930058.2876506}
Amol Deshpande, Lisa Hellerstein, and Devorah Kletenik.
\newblock Approximation algorithms for stochastic submodular set cover with
  applications to boolean function evaluation and min-knapsack.
\newblock {\em ACM Trans. Algorithms}, 12(3):42:1--42:28, April 2016.

\bibitem{goemans:stochasticcovering}
Michel Goemans and Jan Vondr{\'a}k.
\newblock Stochastic covering and adaptivity.
\newblock In Jos{\'e}~R. Correa, Alejandro Hevia, and Marcos Kiwi, editors,
  {\em LATIN 2006: Theoretical Informatics}, pages 532--543, Berlin,
  Heidelberg, 2006. Springer Berlin Heidelberg.

\bibitem{Golovin:2011:AST:2208436.2208448}
Daniel Golovin and Andreas Krause.
\newblock Adaptive submodularity: Theory and applications in active learning
  and stochastic optimization.
\newblock {\em J. Artif. Int. Res.}, 42(1):427--486, September 2011.

\bibitem{hellerstein:jair}
Lisa Hellerstein and Devorah Kletenik.
\newblock Research note: Revisiting the approximation bound for stochastic
  submodular cover.
\newblock To appear in the Journal of Artificial Intelligence Research (JAIR).

\bibitem{johnson:setcover}
David~S. Johnson.
\newblock Approximation algorithms for combinatorial problems.
\newblock {\em J. Comput. Syst. Sci.}, 9(3):256--278, 1974.

\bibitem{Liu:2008:NAS:1376616.1376633}
Zhen Liu, Srinivasan Parthasarathy, Anand Ranganathan, and Hao Yang.
\newblock Near-optimal algorithms for shared filter evaluation in data stream
  systems.
\newblock In {\em Proceedings of the 2008 ACM SIGMOD International Conference
  on Management of Data}, SIGMOD '08, pages 133--146, New York, NY, USA, 2008.
  ACM.

\bibitem{lovasz:setcover}
L.~Lov\'{a}sz.
\newblock On the ratio of optimal integral and fractional covers.
\newblock {\em Discrete Math.}, 13(4):383--390, January 1975.

\bibitem{stein:setcover}
S.K Stein.
\newblock Two combinatorial covering theorems.
\newblock {\em Journal of Combinatorial Theory, Series A}, 16(3):391--397, may
  1974.

\bibitem{Vazirani:2010:AA:1965254}
Vijay~V. Vazirani.
\newblock {\em Approximation Algorithms}.
\newblock Springer Publishing Company, Incorporated, 2010.

\end{thebibliography}
\bibliographystyle{plain}
\end{document}